\documentclass[pra, amsmath, amssymb, twocolumn, floatfix, superscriptaddress]{revtex4-2}
\usepackage{enumitem}
\usepackage{float}
\usepackage{graphicx}
\usepackage{dcolumn}% Align table columns on decimal point
\usepackage{bm}
\usepackage{hyperref}% 
\usepackage[separate-uncertainty=true]{siunitx}
\usepackage{xcolor}
\usepackage{pagecolor}
\usepackage{makecell}
\usepackage{physics}
\usepackage[normalem]{ulem} %For striking things out with \sout
\usepackage{float}
\usepackage{bbold}
\usepackage{notoccite}
\usepackage{tikz}
\usetikzlibrary{automata} 	% Library to draw automata
\usetikzlibrary{arrows} 	% Library for differnt types of arrows

\usetikzlibrary[shapes.arrows]
\usetikzlibrary{shapes.geometric}
\usetikzlibrary{backgrounds}
\usetikzlibrary{positioning}
\usetikzlibrary{calc}
\usetikzlibrary{intersections}
\usetikzlibrary{fadings}
\usetikzlibrary{decorations.footprints}
\usetikzlibrary{decorations.pathmorphing}
\usetikzlibrary{patterns}
\usetikzlibrary{shapes.callouts}
\usetikzlibrary{fit}

\tikzset{->, >=stealth', shorten >=1pt, auto, node distance=1cm, semithick, baseline=(current bounding box.center)}

\usepackage{amsthm}
\newtheorem{theo}{Theorem}

\newtheorem{remark}[theo]{Remark}

\newtheorem{lemma}[theo]{Lemma}

\newtheorem{prop}[theo]{Proposition}
\newtheorem{corollary}[theo]{Corollary}

\newcommand{\id}{\mathbb{1}}
\newcommand{\N}{\mathbb{N}}
 
 \makeatletter
\newcommand{\globalcolor}[1]{%
  \color{#1}\global\let\default@color\current@color
}
\makeatother

\makeatletter
\def\p@subsection{}
\makeatother

\definecolor{matrixgreen}{HTML}{00ff41}
\definecolor{lg}{RGB}{200,200,200}
\definecolor{breezedarkblue}{HTML}{212946}
\definecolor{cybergreen}{HTML}{00ff41}
\definecolor{cyberpink}{HTML}{FE53BB}

\newcommand{\figref}[1]{Fig.~\ref{#1}}
\newcommand{\eref}[1]{Eq.~\ref{#1}}

\newcommand{\ben}[1]{{\color{purple}#1}}

\begin{document}

\title{A universal quantum rewinding protocol with an arbitrarily high probability of success}
%\thanks{Thank you note here}

%\author{H. G. Wells}
%\author{J. Verne}
%\author{I. Asimov}\email{my@email.com}
%\author{A. C. Clarke}

%\affiliation{University of Vienna, Faculty of Physics, Boltzmanngasse 5, A-1090 Vienna}
%\affiliation{Institute for Quantum Optics and Quantum Information (IQOQI), Boltzmanngasse 3, A-1090 Vienna}

\author{D. Trillo}
\author{B. Dive}
\author{M. Navascués}

\affiliation{Institute for Quantum Optics and Quantum Information (IQOQI) Vienna, Boltzmanngasse 3, A-1090 Vienna}

\date{\today}

\begin{abstract} %It's too long

    We present a universal mechanism that, acting on any target qubit, propagates it to the state it had T time units before the experiment started. This protocol works by setting the target on a superposition of flight paths, where it is acted on by uncharacterized, but repeatable, quantum operations. Independently of the effect of each of these individual operations on the target, the successful interference of the paths causes it to leap to its past state. We prove that, for generic interaction effects, the system will reach the desired state with probability 1 after some finite number of steps.
    
\end{abstract}

\maketitle

    \begin{quote}
        \emph{If Time is really only a fourth dimension of Space, why is it, and why has it always been, regarded as something different? And why cannot we move in Time as we move about in the other dimensions of Space?}
        
        H.G. Wells, the Time Machine.
    \end{quote}

    %Time dilation is one of the most surprising predictions of special and general relativity. Take a spaceship and set it on a round trip that crosses a region with a strong gravitational field. When the ship returns, its occupants, who will look younger than expected, will claim that the trip actually took fewer time than measured on Earth. By acting on the ship's motion degree of freedom, we have slowed the evolution of its internal degrees of freedom down.
    
    Time dilation is one of the most surprising predictions of special and general relativity. Take a spaceship and set it on a round trip at very high speeds. When the ship returns, its occupants, who will look younger than expected, will claim that the trip actually took less time than measured on Earth. By acting on the ship's motion degree of freedom, we have slowed down the evolution of its internal degrees of freedom. 
    
    Astounding as it sounds, relativistic time dilation is, nevertheless, extremely impractical: one needs enormous amounts of energy or close proximity to a black hole to observe significant effects. Moreover, regarded as a class of time translations, relativistic time warp is \emph{limited}. One can use it to decelerate the flow of time, but not to invert it or accelerate it.
    
    Some of such limitations disappear when we abandon the realm of relativistic classical physics and enter the realm of non-relativistic quantum mechanics. In prior work~\cite{Navascues2017},~\cite{Trillo2019}, we devised what one could call \emph{universal time translation protocols}: heralded physical processes with the property to decelerate, accelerate and even reverse the evolution of the quantum systems where they act. Notably, these processes do not depend on the free Hamiltonian of the target systems or their Hamiltonian interactions with the ancillary systems used to influence them; they just depend on the Hilbert space dimension of the target. This means that the exact same physical process applied to the spin of an electron, an oscillating kaon or the polarization of a photon (all of them two-dimensional), will translate in time each of these very different systems by the same amount. 
    
    All universal processes for time translation are necessarily probabilistic: this is so because they are expected to be sound even when the target does not interact at all with the experimental device, in which case the probability of a successful rewinding cannot be other than zero. What happens, though, when the target only interacts weakly with the experimental device? In all known universal rewinding protocols, one observes that the probability of success, while strictly greater than zero, becomes impractically small. This makes the implementation of such processes (see, e.g., \cite{Gong2019, Li2019}) a mere academic exercise and raises the question of whether quantum physics allows for efficient universal rewinding protocols at all.

    \begin{figure*}[!t]
        \begin{tabular}{cc|ccc|cc}
             (a) \includegraphics[scale = 0.7]{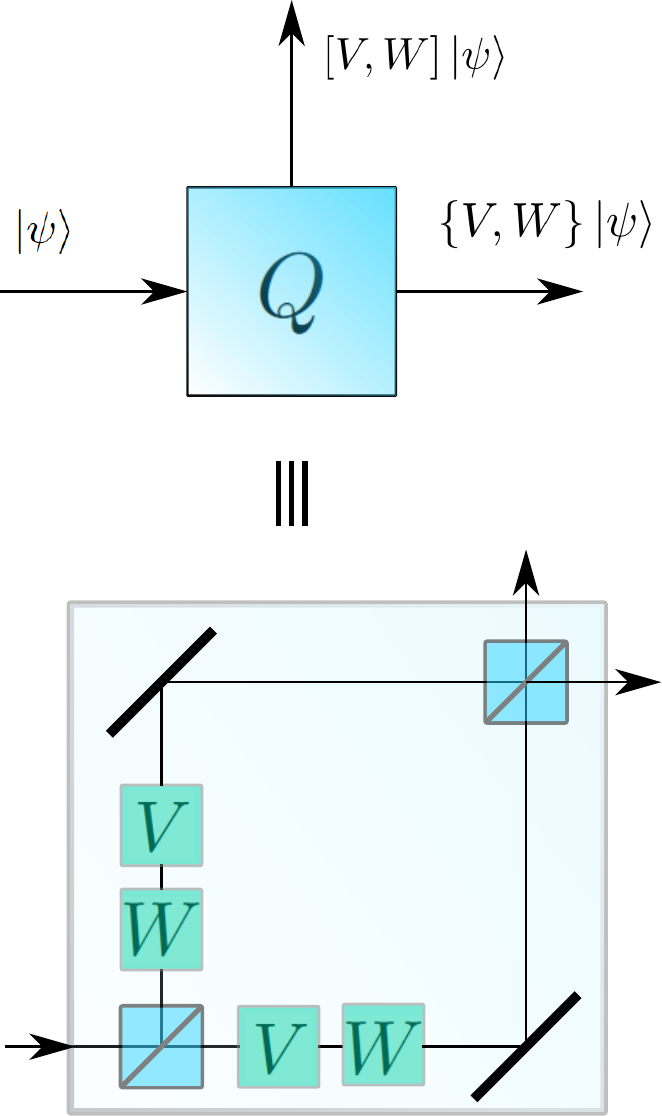} & \hspace{5mm} & \hspace{5mm} &
             (b) \includegraphics[scale = 0.7]{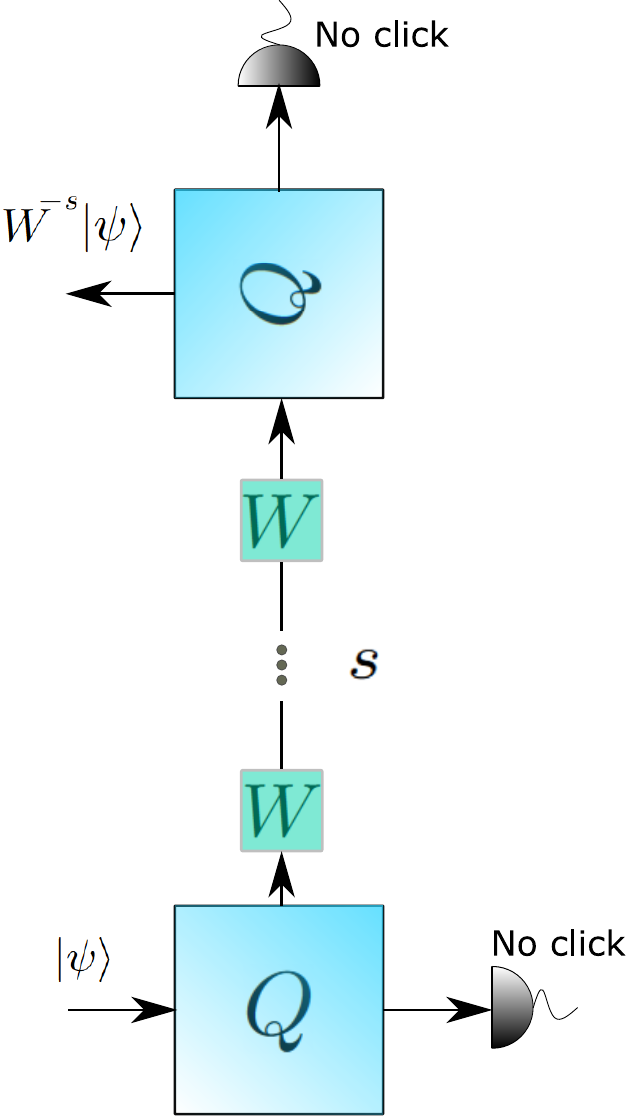} & \hspace{5mm} &  & (c) \includegraphics[scale=0.7]{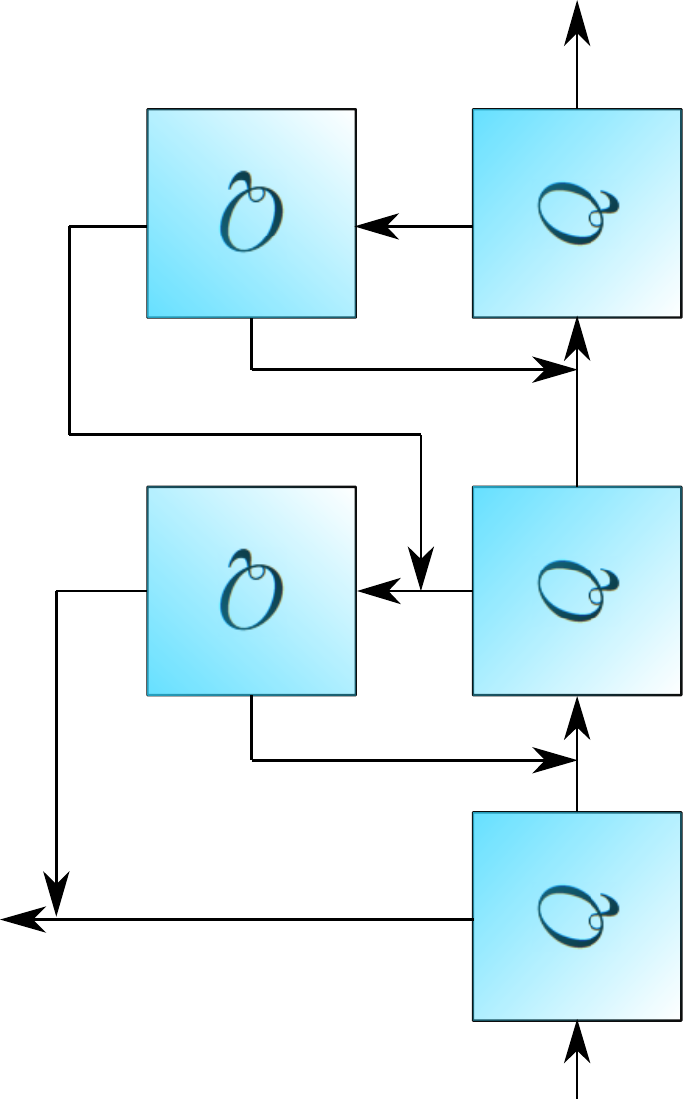}
        \end{tabular}
        \caption{(a) The gate $Q$ specified by Eq.(\ref{building_block}). This quantum operation, the building block of the whole protocol, can be implemented in different ways depending on the physical systems under consideration. The figure shows a possible way to implement $Q$ on photons with an interferometer. (b) A schematic of the full protocol without corrections. (c) A schematic of a level 2 correction, rotated 90 degrees counterclockwise for convenience. This whole figure substitutes each $Q$ in part (b). Following any path in the figure that ends up in the bottom left outputs a state proportional to $[V,W] \ket{\psi}$. By adding more $Q$ gates in the same fashion the success of the scheme becomes more and more probable.}
        \label{fig:qdiagram}
    \end{figure*}  
    
    In this paper, we present a universal rewinding protocol for two-level systems that, in the generic case, rewinds the target with certainty. In analogy with the relativistic case, the protocol requires manipulating the target system's motion degree of freedom; the main difference is that we allow setting the target on a superposition of different flight paths. For some of these paths, the evolution of the target's internal two-level system might not be guided by its (unknown) free Hamiltonian $H_0$, but by some other (also unknown) interacting Hamiltonian $H_I$, e.g., due to the presence of a magnetic field along the path. We make these paths sequentially interfere in non-trivial ways until the system reaches a terminal interference pattern. When that happens, the internal degree of freedom of the target is guaranteed to have experienced the desired rewinding. 
    
    For generic $H_0, H_I$, we prove that, with probability $1$, the system will reach a terminal interference pattern in some finite number of steps. More generally, let $\|[e^{-iH_0\Delta T},e^{-iH_I\tau}]\|>\epsilon$, for some $\Delta T,\tau,\epsilon>0$. Then, for any $0<q<1$, we can find a number $m$ such that, with probability at least $q$, the rewinding protocol will reach a terminal interference pattern in $m$ time-steps or fewer. Any such $m$-trimmed protocol, if successful, will rewind a target by $T$ time units in time at most $T+m(\tau+\Delta T)$, which makes the protocol's running time asymptotically optimal on $T$~\cite{Trillo2019}. The structure of the general rewinding protocol makes it amenable for any experimental setup that allows the target to be acted upon by a beam splitter, and, in fact, a $2$-trimmed version of the protocol has already been successfully implemented with photonic systems \cite{experimental}.

    There exist traditional methods to rewind a quantum system with an unknown free Hamiltonian (such as the refocusing techniques used in nuclear magnetic resonance~\cite{refocusing}), as well as modern methods based on the notion of quantum combs \cite{Quintino2018}. Those protocols are of the form: `act on the target system with the $Z$ Pauli gate, let the system freely evolve for a microsecond, then act on it with the $X$ Pauli gate', etc. They hence require the ability to implement specific operations on the target system (in this case, the Pauli gates). Thus, the physical process used to refocus a specific two-level system, say, the spin of an electron, by means of, e.g., a sequence of magnetic fields, will fail to rewind some other two-level systems, such as the polarization of a photon, or even another, more distant, electron. In contrast, the universal protocols detailed in~\cite{Navascues2017,Trillo2019}, and also in this paper, just require the capacity to interact with the target system in \emph{some} way. The specific effect of this interaction on the target (e.g.: to propagate it through the Hadamard gate) just affects the probability of the heralded success of the process, but not its ultimate outcome, namely, the rewinding of the target. Besides the protocols in~\cite{Navascues2017,Trillo2019}, another example of a universal protocol for time translation can be found in~\cite{sandu}, where the authors combine both quantum theory and general relativity to devise a `time translator', capable of rewinding or fast-forwarding quantum systems. Unfortunately, the probability of success of the time translator is too small to make this device practical or experimentally realisable with current technology.

    %It is worth comparing these results with past works. There exist other quantum information processing proposals to rewind a system with an unknown free Hamiltonian~\cite{refocusing},~\cite{Quintino2018}, ~\cite{Lesovik2019}. These require, however, the ability to implement characterized operations on the target system, and so they are not universal. The work of~\cite{sandu} combines both quantum theory and general relativity to devise a `time translator', capable of rewinding or fast-forwarding quantum systems. While this method time-translates all quantum systems independently of their Hilbert space dimension, it has two drawbacks: (1) it only works approximately, and under a restriction on the free Hamiltonian of the target; (2) if we demand a reasonable precision, the probability of success of the process becomes astronomically small. 

\vspace{10pt}    
\noindent\emph{The Protocol}
    
    In the following, given a  time unit $\Delta T>0$, we introduce a universal physical process that rewinds any two-level quantum system by any amount $T=s\Delta T$, where $s$ is an arbitrary natural number. This process, acting on a target system with free Hamiltonian $H_0$, will propagate the target's initial quantum state by $e^{iH_0s\Delta T}$, thus leaving the target on the state it had $s\Delta T$ time units before the experiment started.
    
    The basic building block of the protocol is the gate $Q$ depicted in \figref{fig:qdiagram}. Denoting by $\ket{\psi}$ the state of the target, this gate performs the transformation
    \begin{equation}
        Q\ket{\psi}\ket{\rightarrow}\propto [V,W]\ket{\psi}\ket{\uparrow}+\{V,W\}\ket{\psi}\ket{\rightarrow}.
        \label{building_block}
    \end{equation}
    \noindent Here $W:=e^{-iH_0\Delta T}$ and $V$ denotes an unknown unitary map, detailed below. The kets $\ket{\rightarrow}, \ket{\uparrow}$ respectively label a left-to-right and a bottom-up trajectory of the target system, as seen in the figure when the letter Q is upright. If, right after implementing $Q$, we measure the target's motion degree of freedom in the $\{\ket{\uparrow},\ket{\rightarrow}\}$ basis, the target will be propagated by either $\{V,W\}$ or $[V,W]$, depending on the measurement result.
    
    We next dedicate some lines to explain how to universally realize the gate $Q$, also known in the literature as \texttt{SWITCH} \cite{Chiribella}, for some uncharacterized matrix $V$. Let $O$ be a repeatable physical operation (e.g.: switching on a magnetic field, releasing an electron) of duration $\tau$, whose effect on the target's internal degree of freedom is to propagate its ket by some unknown operator $V$. This can be achieved by, e.g., switching on an interacting Hamiltonian for some time $\tau$, or, as described in \cite{Navascues2017, Trillo2019}, by making the target unitarily interact with a probe, which is post-selected onto a given pure state after the interaction. Note that $V$ will be a unitary matrix in the first case and non-unitary in the second.
    
    Given the ability to conduct any such operation $O$, one can implement the gate $Q$ by playing with the motion degree of freedom of the target: it suffices to put the latter in an equal superposition of two paths. In the first path, the target is allowed to evolve freely for time $\Delta T$ and then we act on it with $O$ for time $\tau$. In the second path, we first act on the target with $O$ for time $\tau$ and then we let it evolve freely for time $\Delta T$. The state of the target at this stage will thus be proportional to $VW\ket{\psi}\ket{\gamma_1}+WV\ket{\psi}\ket{\gamma_2}$, where $\gamma_1, \gamma_2$ denote the two trajectories. Next we make the two trajectories interfere, by conducting the unitary operation
    \[
    \ket{\gamma_1}\mapsto \frac{1}{\sqrt{2}}\left(\ket{\rightarrow}-\ket{\uparrow}\right), \hspace{5mm} \ket{\gamma_2}\mapsto \frac{1}{\sqrt{2}}\left(\ket{\rightarrow}+\ket{\uparrow}\right)
    \] (in optical systems, this can be achieved with a balanced beam splitter), arriving at \eref{building_block}. 
    
    The gist of the protocol is to apply the gate $Q$ over and over to the target until it reaches a state proportional to $W^{-s}\ket{\psi}=e^{iH_0 s\Delta T}\ket{\psi}$. To achieve this goal, we rely on three general properties of $2\times 2$ matrices.

    \begin{prop}
    \label{fundamental_prop}
    Let $V,W$ be arbitrary $2\times 2$ matrices, and define $x \equiv [V, W], y \equiv \{V, W\}$. Then we have that: 
    \begin{enumerate}[label=(\alph*)]
     \item 
    $x^2\propto\id_2$.
    \item 
    If $W$ is invertible, then, for any natural number $s$, $xW^sx\propto W^{-s}$.
    \item 
    For any natural number $n$, $y^nxy^n\propto x$. 
    \end{enumerate}
    \end{prop}
    \noindent We remark that the proportionality factors on the right-hand sides of equations (a-c) are functions of the entries of the matrices $V,W$, and might vanish for some values of $V, W$. The reader can find a  proof of Proposition \ref{fundamental_prop} at the end of the letter.
    
    Proposition \ref{fundamental_prop} suggests a simple method to bring the target system to state $W^{-s}\ket{\psi}$ through consecutive uses of gate $Q$. First, we aim to effect the transformation $\ket{\psi}\to x\ket{\psi}$. Once there, all we have to do is wait for time $s\Delta T$ and manage to enforce the transformation $x$ once more. The final state will then be $xW^sx\ket{\psi}$, that, by relation (b), is proportional to the state $W^{-s}\ket{\psi}=e^{i H_0 s \Delta T}\ket{\psi}$. In that case, the target will have been translated by $-s\Delta T$ time units.
    
    The shortest way to rewind the system hence requires two applications of gate $Q$, see \figref{fig:qdiagram}~(b). Provided that the target system emerges from gate $Q$ through its vertical output port, the system will have been acted upon by $x$. Next, we wait for time $s$ and then we input the system in $Q$ again. If, once more, the target exits the gate through its vertical port, then we can guarantee that the rewinding process took place.
    
    It could happen, though, that the target exits the first gate through its horizontal port. In that case, the system will be propagated by $y$ instead of $x$. To proceed with the rewinding protocol, we must eliminate this operator. A possible path out is given by taking $n=1$ in relation (c), namely, by the identity $yxy\propto x$. It follows that, if we make the system pass through two more $Q$ gates and it exits the first one through the vertical port; and the second one, through the horizontal port, the system will end up in a state proportional to $x\ket{\psi}$. The situation is thus the same as if the target had exited through the vertical port in the original $Q$ gate, see \figref{fig:qdiagram}~(c). Hence we can wait for $s\Delta T$ time units before trying to effect another transformation $x$ on the system.

    By virtue of relations (a), (b), (c) in Proposition \ref{fundamental_prop}, whichever sequence of ports the system happens to exit will propagate the target by an operator of the form $xy^n$ or $y^n$. In the first case, $n$ consecutive exits through the horizontal port of gate $Q$ will propagate the system by $x$. In the second one, a vertical detection, followed by $n$ consecutive horizontal ones, will have the same effect. Hence, no matter how advanced the protocol is, there always exists a chance of bringing the target to the terminal configuration $xW^{s}x\ket{\psi}$, as sketched in \figref{fig:qdiagram}~(c). Note that relations (a)-(c) hold even if the matrices $V,W$ are not unitary. The protocol can thus be used, e.g., to rewind a two-level system undergoing a continuous decay governed by a non-Hermitian Hamiltonian, such as a neutral kaon \cite{kaons}. 
    
    Notice as well that, should we enforce any limit $m$ on the number of times that gate $Q$ can be applied, the running time of the protocol would be upper bounded by $T'=m(\Delta T+\tau) + s\Delta T$ (recall that $\tau$ is the time it takes to implement the operation $O$) On the other hand, the protocol, if successful, would rewind the target system by an amount $T=s\Delta T$. Hence $T'=T+O(1)$ and, by \cite{Trillo2019}, this implies that such an `$m$-trimmed' universal rewinding protocol runs on (asymptotically) minimal time. 
    
    It remains to be seen how likely it is that the (trimmed or untrimmed) protocol succeeds. In principle\ben{\sout{,}} it could be that, even allowing an unlimited number of uses of gate $Q$, the system never reaches a terminal configuration $xW^sx\ket{\psi}$. In this regard, note that, if the physical operation $O$ has no effect whatsoever on the target (namely, if $V=e^{-iH_0\tau}$), then the latter will keep evolving unperturbed, no matter how many times we act on it with the $Q$ gate. More generally, one can see that the rewinding protocol will fail with certainty whenever $[V,W]=0$.
    
    %The condition $[V,W]=0$ requires, however, a high degree of fine-tuning, if the experimental setup can indeed influence the target system, and will be violated by generic interactions $V,W$. One therefore wonders what the chances of success are when the operation $O$ is actually perturbing the system of interest, i.e., when $[V,W]\not=0$.
    
    The condition $[V,W]=0$, violated by generic interactions $V,W$, requires a high degree of fine-tuning if the experimental setup is capable of perturbing the target system at all. One therefore wonders what the chances of success are when $[V,W]\not=0$.
    
    Using techniques from probability theory \cite{martingales}, we prove in the Appendix that, provided that $V,W$ are unitary and $[V,W]\not=0$, the target will reach the pattern $xW^sx\ket{\psi}$ after a random finite number of uses of $Q$ with probability $1$. Moreover, given a lower bound on $\|[V,W]\|$, we show how to compute, for any $0<q<1$, a finite number $m$ for which the corresponding $m$-trimmed protocol is successful with probability at least $q$.

    \begin{figure}[H]
        \centering
        \begin{tabular}{c}
             (a) \includegraphics[scale = 0.65]{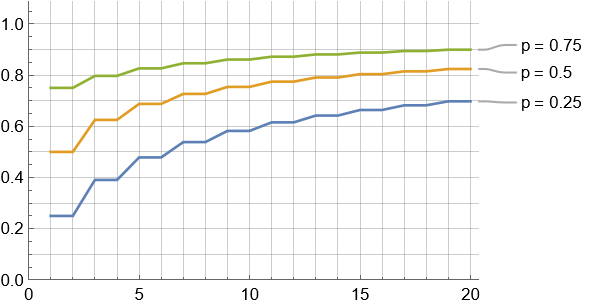} \\
             \\
             \\
             (b) \includegraphics[scale = 0.65]{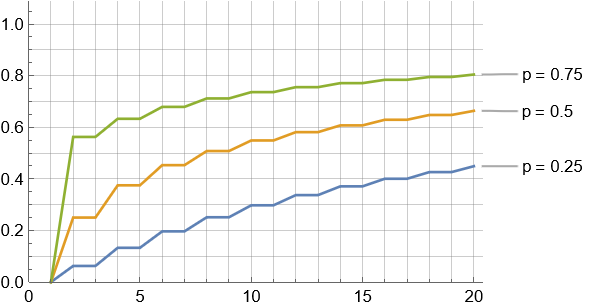}
        \end{tabular}
        \caption{The protocol's probability of success (vertical axis) as a function of the number $m$ of uses of gate $Q$ (horizontal axis), for different values of the probability $p$ of exiting the $Q$ gate through the vertical port (as proven in the Appendix, $p$ depends on $V,W$, but not on the state of the target).
        (a) The probability of successfully implementing $[V,W]$ (b) The probability of successfully rewinding the system for the full adaptive protocol. Note that this is not the square of (a).}
        \label{fig:probability}
    \end{figure}  
    
    \vspace{10pt}
\noindent\emph{Conclusion}

    In this work we have presented a universal time-rewinding mechanism for two-level quantum systems. Unlike previously proposed protocols, ours can reach an arbitrarily high probability of success and is asymptotically optimal in the time required to perform the rewinding operation, thus answering the question of whether such processes are allowed by the laws of quantum mechanics. Since the process is universal, the experimenter performing the protocol does not need any knowledge about the target system, its internal dynamics or even the specifics of the interaction between the system and the experimenter (other that there is one). 
    
    An immediate open question is whether our results can be generalized to higher dimensions. For any $d$, are there universal rewinding mechanisms for $d$-dimensional systems with an arbitrarily high probability of success? Or, on the contrary, does there exist a pair of $d$-dimensional unitaries $V,W$ for which universal rewinding can be achieved, but with a bounded probability of success? 
    
    Another topic worth exploring is whether our techniques can be adapted to devise more efficient time-translation mechanisms for \emph{ensembles} of two-level systems. As shown in \cite{Trillo2019}, there exist universal mechanisms which, acting on an ensemble of $n$ identical $d$-level systems for time $T+O(1)$, have the effect of time-translating one of the systems $nT$ time units to the future, while leaving the remaining $n-1$ systems in the state they had at the beginning of the experiment (a phenomenon we call time transfer). Even for two-dimensional systems, known time transfer mechanisms typically exhibit very low probabilities of success \cite{Trillo2019}; so low, in fact, that their experimental implementation is out of the question. The relations (a)-(c) identified for $2\times 2$ matrices, together with methods from probability theory, might help overcome this difficulty.

\vspace{10pt}       
\noindent\emph{Acknowledgements}

    D.T. is a recipient of a DOC Fellowship of the Austrian Academy of Sciences at the Institute of Quantum Optics and Quantum Information (IQOQI), Vienna. D.T. thanks Claudio Procesi for providing a shorter proof of Lemma \ref{lemma_1}, and Edgar A. Aguilar and David Martínez-Rubio for useful discussions about martingales and probability theory. This work was not funded by the European Research Council.

\vspace{10pt}

\noindent\emph{Proof of Proposition \ref{fundamental_prop}}

    The first ingredient of the proof is the following lemma:
     
    \begin{lemma} \label{lemma_1}
        Let $V,W$ be $2 \times 2$ matrices and $n\in \N \cup \{0\}$. Then, $\tr([V,W]\{V,W\}^n) = 0$.
    \end{lemma}
    \begin{proof}
    By the Caley-Hamilton theorem, for $n\geq 2$, the $2\times 2$ matrix $\{V,W\}^n$ is a linear combination of $\id, \{V,W\}$. Hence, it is enough to show that the lemma holds for $n=0,1$, and this is a simple consequence of the cyclicity of the trace.
    \end{proof}
    
    Now, note that any $2 \times 2$ traceless matrix can be written as a linear combination of the Pauli matrices $\sigma_X, \sigma_Y, \sigma_Z$, and thus its square is proportional to the identity matrix. This applies to the commutator $[V,W]$, the matrix polynomial $W^s[V,W]$ and, by Lemma \ref{lemma_1}, to $[V,W]\{V,W\}^n$. We thus have that, for all $2\times 2$ matrices,
    
    \begin{align}
        &[V,W]^2\propto \id,\nonumber\\
        &W^s[V,W]W^s[V,W]\propto\id,\nonumber\\
        &[V,W]\{V,W\}^n[V,W]\{V,W\}^n\propto \id.
        \label{consequences}
    \end{align}
    
    If $W$ is invertible, then we can multiply the second expression by $W^{-s}$ on the left and arrive at the identity
    
    \begin{equation}
        [V,W]W^s[V,W]\propto W^{-s}.
        \label{reverse}
    \end{equation}
    
    Similarly, multiplying the third line of \eref{consequences} by $[V,W]$ on the left and invoking the first line, we arrive at 
    
    \begin{equation}
        \{V,W\}^n[V,W]\{V,W\}^n\propto [V,W].
        \label{complicated}
    \end{equation}
    
    Note that the last step is only rigorous if the proportionality factor in the expression $[V,W]^2\propto\id$ is non-zero. As it turns out, by the Cayley-Hamilton theorem, this factor is $-\mbox{det}([V,W])$. Let us then prove that the relation also holds for $\mbox{det}([V,W])=0$. 
    
    Define the matrices $x \equiv [V,W]$, $y \equiv \{V,W\}$ and $z \equiv y^nxy^n$. The matrices $x$ and $z$ have in this case rank at most $1$. Since both $x,z$ have zero trace ($z$, by virtue of Lemma \ref{lemma_1}), it follows that one can write them as $x=\lambda \ketbra{\phi}{\phi^\perp}$, $z=\nu\ketbra{\varphi}{\varphi^\perp}$, where $\braket{\phi}{\phi^\perp}=\braket{\varphi}{\varphi^\perp}=0$. Now, by the third line of \eref{consequences}, $xz\propto\id$. Since the left-hand side of the relation has rank at most one, it follows that $xz=0$. This is only possible if $\nu=0$ (note that $\lambda=0$ implies $\nu=0$), in which case $z=0\propto x$; or if $\lambda,\nu\not=0$ and $\braket{\phi^\perp}{\varphi}=0$, from which $z\propto x$. In either case, relation (\ref{complicated}) holds.
    
    The first line of \eref{consequences}, \eref{reverse} and \eref{complicated} are, respectively, the $2\times 2$ matrix relations (a), (b), (c) claimed to hold in Proposition \ref{fundamental_prop}. This finishes the proof.

%\nocite{*}

\bibliographystyle{unsrt}
\bibliography{bibliography}
\appendix

\section{Completeness of the rewinding protocol}\label{appendix: the protocol}

    In this section we prove that the rewinding protocol sketched in the main text is complete, in the sense that, provided that $V,W$ are unitary matrices and $[V,W]\not=0$, the protocol will always halt. 
    
    The first step is to prove that, for $V, W$ unitary (and $2\times 2$), the probability that a quantum state $\ket{\psi}$ exits the horizontal port of gate $Q$ is independent of $\ket{\psi}$. To see this, we apply the Cailey-Hamilton theorem to the $2\times 2$ matrix $VWV^\dagger W^\dagger$, obtaining
    
    \begin{align}
        &(VWV^\dagger W^\dagger)^2-\mbox{tr}(VWV^\dagger W^\dagger)VWV^\dagger W^\dagger+\nonumber\\
        &\mbox{det}(VWV^\dagger W^\dagger)=0,
    \end{align}
    The last determinant equals $1$. Thus, Multiplying the above expression by $W V$ on the left and by $W^\dagger V^\dagger$ on the right, we find that
    
    \begin{equation}
        V^\dagger W^\dagger VW+ W^\dagger V^\dagger WV=\mbox{tr}(VWV^\dagger W^\dagger)\id.
        \label{constant_prob}
    \end{equation}
    
    \noindent Now, the probability $p$ that a state $\ket{\psi}$ leaves the $Q$ gate through the horizontal port is 
    
    \begin{align}
        &\frac{1}{4}\bra{\psi}[V,W]^\dagger [V,W]\ket{\psi}=\nonumber\\
        &\frac{1}{4}\bra{\psi}2\id - V^\dagger W^\dagger VW- W^\dagger V^\dagger WV\ket{\psi}.
    \end{align}
    By \eref{constant_prob}, the last expression just depends on the invariant $\mbox{tr}(VWV^\dagger W^\dagger)$ and not on the state itself.
    
    The independence of the quantum state of the probability of exiting the $Q$ gate from either port, together with relations (a)-(c) from Proposition $1$ in the main text, allows us to model the evolution of the target state, as it is acted sequentially with the gate $Q$, as a classical particle undergoing a random walk in the directed graph shown in \figref{fig: graph}.
    
    \begin{figure}[H]
        \centering
        \begin{tikzpicture}[state/.style={circle, draw, minimum size=1cm}]
         \node[state] (0down) {$\underline{0}$};
         \node[state] (1down) [right of = 0down, node distance = 2cm] {$\underline{1}$};
         \node[] (down) [right of = 1down, node distance = 2cm] {$\cdots$};
         \node[state] (0up) [above of = 0down, node distance = 2cm] {$\overline{0}$};
         \node[state] (1up) [above of = 1down, node distance = 2cm] {$\overline{1}$};
         \node[] (up)[above of = down, node distance = 2cm]{$\cdots$};
         \node[state](-1down) [left of = 0down, node distance = 2cm]{$-\underline{1}$};
         \node[state](-1up)[above of = -1down, node distance = 2cm]{$-\overline{1}$};
         \node[](left)[left of = -1down, node distance=2cm]{$\cdots$};
         \node[](leftup)[left of =-1up, node distance=2cm]{$\cdots$};

         \draw[->] (0down) to node {} (1down);
         \draw[-] (0down) to node {} (0up);
         \draw[->] (1down) to node {} (down);
         \draw[-] (1down) to node {} (1up);
         \draw[->] (1up) to node {} (0up);
         \draw[->] (-1down) to node {} (0down);
         \draw[-] (-1down) to node {} (-1up);
         \draw[->] (0up) to node {} (-1up);
         \draw[->] (up) to node {} (1up);
         \draw[->] (left) to node {} (-1down);
         \draw[->] (-1up) to node {} (leftup);
        
        \end{tikzpicture}
    
        \caption{\textbf{A random walk modelling the word problem.} In this graph, starting in state $\underline{0}$ at $t=0$, at each discrete step the classical particle moves in the vertical direction with probability $p$ and in the horizontal direction with probability $1-p$. The goal is to get back to $\underline{0}$ at a positive time. Consequently, a move in the horizontal direction corresponds to the operation $\ket{\psi} \mapsto y \ket{\psi}$, and a move in the vertical direction to the operation $\ket{\psi} \mapsto x \ket{\psi}$.}
        \label{fig: graph}
    \end{figure}
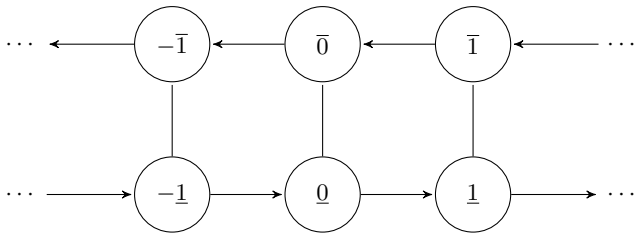
    
    At each time-step, the particle can move in the horizontal direction with probability $1-p$: this corresponds to the quantum target leaving a $Q$ gate through the horizontal port, hence propagating the current quantum state by the operator $y$. Alternatively, with probability $p$, the classical particle will move in the vertical direction of the graph. This corresponds to the target leaving the $Q$ gate through the vertical port, which propagates its state by $x$. If the initial position of the classical particle is $\underline{0}$, then, by the time the particle reaches $\overline{0}$, the quantum target system will have been propagated by $x$. This is independent of the graph path taken by the classical particle, by virtue of relations (a)-(c) in the main text.
    
    Once the classical particle is in $\overline{0}$, we would stop the random walk momentarily and let the quantum target system evolve freely for $s\Delta T$ time units. Then we would act again with the $Q$ gate on the system, thus continuing the random walk until the classical particle arrives at $\underline{0}$, at which point the target has been propagated by $xW^{-s} x$ and hence it would have been rewound.
    
    We will next prove that the probability that the classical particle passes through $\overline{0}$ and then $\underline{0}$ is $1$ provided that $p\not=0$. Note that any trajectory that ends up again in $\underline{0}$ will necessarily pass through $\overline{0}$, so we just need to compute the probability that the classical particle makes a closed trajectory. Since the waiting time does not affect the probability of success of the protocol, we set $s=0$ for the remainder of the discussion.
    
    Let $\{S_n\}_{n\geq 0}$ be the sequence of random variables which describe this random walk when starting in the state $S_0$. We define the hitting time of node $a$ from state $b$ as the random variable
    \[
    T_{b\rightarrow a}:=\inf\{n>0 \hspace{1mm} | \hspace{1mm} S_n= a, S_0=b\},
    \]
    which takes values in $\mathbb{N}\cup \{+\infty\}$. The probability of successfully finishing the protocol in $m$ steps is therefore given by $P(T_{\underline{0}\rightarrow \underline{0}}=m)$. To calculate this probability we will exploit the symmetries of the graph. For the reader not familiar with some of the concepts used in the proof, we recommend~\cite{martingales}.
    
    \begin{theo} \label{theorem}
        For all natural $m\geq 1$ and all $p\in [0,1]$, with the convention that $0^0=1$, 
        \begin{align*}
            &P\left(T_{\underline{0}\rightarrow \overline{0}}= 2m-1\right) =\\
             &=\sum_{n=1}^m (-1)^{n+1}\binom{1/2}{n} \binom{1-2n}{m-n} (2p)^{2n-1}(2p-1)^{m-n},\\
            &P\left(T_{\underline{0}\rightarrow \overline{0}}= 2m\right) =0.
        \end{align*}
    \end{theo}
    
    \begin{remark} \label{p=1/2}
        The use of $0^0=1$ in Theorem \ref{theorem} is justified in two ways. First, when $p\neq 1/2$ there is no indeterminacy in the formula, which is a continuous function of $p$. Taking the limit $p\rightarrow 1/2$ we obtain
        \begin{equation} \label{symmetric random walk}
            P\left( T_{\underline{0}\rightarrow \overline{0}}= 2m-1 \right)= (-1)^{m+1} \binom{1/2}{m},
        \end{equation}
        so regarding $0^0$ as $1$ is the natural choice to make the probability continuous on $p$. Furthermore, equation \eref{symmetric random walk} is also the probability for the hitting time of the origin in a simple symmetric random walk of the integers starting at the origin. Closer inspection of our graph reveals that indeed, when $p=1/2$ (and only in this case), these two processes are equivalent for the purposes of this random variable. This is therefore the correct formula for $p=1/2$.
    \end{remark}
    
    We proceed to prove Theorem \ref{theorem}.
    
    \begin{proof}
        
        We define the generating function
        \[
            f\left(\alpha\right):= \mathbb{E}\left[\alpha^{T_{\underline{0}\rightarrow \overline{0}}}\right] = \sum_{n=1}^\infty  P\left(T_{\underline{0}\rightarrow \overline{0}}=n\right) \alpha^n.
        \]
        Note that in principle this is only correctly defined for $\alpha < 1$, as the probability of having an infinite hitting time could be non-zero. The correctness of the last equality is justified a posteriori, when we will see that $\lim_{\alpha\rightarrow 1^-}f(\alpha)=1.$ For the time being, assume that $0<\alpha <1$.
        We have
        \begin{align*}
            &f\left(\alpha\right)=\\
            &=  \mathbb{E}\left[\alpha^{T_{\underline{0}\rightarrow \overline{0}}} \hspace{1mm} | \hspace{1mm} S_1= \overline{0} \right] P\left(S_1=\overline{0}\right) + \\
            &+  \mathbb{E}\left[\alpha^{T_{\underline{0}\rightarrow \overline{0}}} \hspace{1mm} | \hspace{1mm} S_1=\underline{1} \right] P\left(S_1=\underline{1}\right)\\
            &= p \mathbb{E}\left[\alpha^{T_{\underline{0}\rightarrow \overline{0}}} \hspace{1mm} | \hspace{1mm} S_1= \overline{0} \right] + \left(1-p\right) \mathbb{E}\left[\alpha^{T_{\underline{0}\rightarrow \overline{0}}} \hspace{1mm} | \hspace{1mm} S_1=\underline{1} \right]\\
            &= p \alpha + \left(1-p\right) \mathbb{E}\left[\alpha^{T_{\underline{0}\rightarrow \overline{0}}} \hspace{1mm} | \hspace{1mm} S_1=\underline{1} \right].
        \end{align*}
        
        However, the graph is invariant under horizontal translations, so  
        \begin{align*}
            P\left(T_{\underline{0}\rightarrow \overline{0}}=n \hspace{1mm} | \hspace{1mm} S_1=\underline{1}\right) &=P\left(T_{\underline{1}\rightarrow \overline{0}}=n-1\right) \\
            &= P\left(T_{\underline{0}\rightarrow -\overline{1}}= n-1 \right).
        \end{align*}
        Therefore,
        \begin{align*}
            \mathbb{E}\left[\alpha^{T_{\underline{0}\rightarrow \overline{0}}} \hspace{1mm} | \hspace{1mm} S_1=\underline{1} \right] &= \sum_{n=2}^\infty P\left(T_{\underline{1}\rightarrow \overline{0}}=n-1 \right) \alpha^n \\
            & = \alpha \mathbb{E}\left[\alpha^{T_{\underline{0}\rightarrow -\overline{1}}} \right].
        \end{align*}
        We can divide the process of getting to $-\overline{1}$ from $\underline{0}$ in two parts: by going for the first time to $\overline{0}$ from $\underline{0}$ and then visiting $-\overline{1}$ from $\overline{0}$ also for the first time. The probabilities are decomposed as follows:
        \begin{align*}
            &P\left(T_{\underline{0}\rightarrow \overline{-1}}=n \right) =\\
            &=P\left( \inf \{k> T_{\underline{0}\rightarrow \overline{0}} \hspace{1mm} | \hspace{1mm}S_k= -\overline{1}, S_{T_{\underline{0}\rightarrow \overline{0}}}=\overline{0},  S_0=\underline{0} \} =n \right)\\
            &=\sum_{m< n} P\left(T_{\underline{0}\rightarrow \overline{0}}=m\right) P\left(\inf\{k>m \hspace{1mm} | \hspace{1mm} S_k= -\overline{1}, S_m=\overline{0} \} = n\right)\\
            &= \sum_{m<n} P\left(T_{\underline{0}\rightarrow \overline{0}}=m\right) P\left(T_{\overline{0}\rightarrow -\overline{1}}=n-m\right)\\
            &= P\left(T_{\underline{0}\rightarrow \overline{0}} + T_{\overline{0}\rightarrow -\overline{1}} =n \right),
        \end{align*}
        where we have used the Markov property in the third step. Note that by the strong Markov property, $\inf \{k> T_{\underline{0}\rightarrow \overline{0}} \hspace{1mm} | \hspace{1mm}S_k= -\overline{1}, S_{T_{\underline{0}\rightarrow \overline{0}}}= \overline{0} \}$ (which has the same distribution as $T_{\overline{0}\rightarrow -\overline{1}}$) is independent of $T_{\underline{0}\rightarrow \overline{0}}$. In particular,
        \[
        \mathbb{E}\left[\alpha^{T_{\underline{0}\rightarrow -\overline{1}}}\right] = \mathbb{E}\left[\alpha^{T_{\underline{0}\rightarrow \overline{0}}}\right] \mathbb{E}\left[\alpha^{T_{\overline{0}\rightarrow -\overline{1}}}\right] = \mathbb{E}\left[\alpha^{T_{\underline{0}\rightarrow \overline{0}}}\right]  \mathbb{E}\left[\alpha^{T_{\underline{0}\rightarrow \underline{1}}}\right],
        \]
        where the last equality follows from the reflection symmetry of the graph.
        
        Repeating the arguments made at the beginning for $T_{\underline{0}\rightarrow \overline{0}}$ we get that
        \[
        \mathbb{E}\left[\alpha^{T_{\underline{0} \rightarrow \underline{1}}}\right] = \alpha p \mathbb{E}\left[ \alpha^{T_{\underline{0}\rightarrow -\overline{1}}} \right]+(1-p)\alpha.
        \]
        Combining everything, 
        \[
        f(\alpha)=p\alpha + \frac{(1-p)^2 \alpha^2 f(\alpha)}{1-pf(\alpha)\alpha}
        \]
        or the second degree equation
        \[
        \alpha p f(\alpha)^2 +(\alpha^2-2p\alpha^2-1)f(\alpha) + p\alpha =0,
        \]
        which has the solutions
        \[
        \frac{1+2p\alpha^2-\alpha^2 \pm \sqrt{(1+2p\alpha^2-\alpha^2)^2-4p^2\alpha^2}}{2p\alpha}.
        \]
        The correct behaviour as $\alpha \rightarrow 0^+$ is obtained with the minus sign in front of the square root, so let us expand this one as a power series on $\alpha$ centered at zero:
        \begin{widetext}
            \begin{align*}
                f(\alpha)
                &= \frac{1+2p\alpha^2-\alpha^2 - \sqrt{(1+2p\alpha^2-\alpha^2)^2-4p^2\alpha^2}}{2p\alpha}\\
                &=\frac{-(1+2p\alpha^2-\alpha^2)\sum_{n=1}^\infty (-1)^n \binom{1/2}{n}\left(\frac{2p\alpha}{1+2p\alpha^2-\alpha^2}\right)^{2n}}{2p\alpha}\\
                &= \sum_{n=1}^\infty (-1)^{n+1}\binom{1/2}{n}\left(\frac{2p\alpha}{1+2p\alpha^2-\alpha^2}\right)^{2n-1}\\
                &= \sum_{n=1}^\infty (-1)^{n+1}\binom{1/2}{n} (2p)^{2n-1} \alpha^{2n-1} \left( \frac{1}{1+(2p-1)\alpha^2}\right)^{2n-1}\\
                &= \sum_{n=1}^\infty (-1)^{n+1}\binom{1/2}{n} (2p)^{2n-1} \alpha^{2n-1} \sum_{k=0}^\infty \binom{1-2n}{k} (2p-1)^k \alpha^{2k}\\
                &=\sum_{m=1}^\infty \sum_{n+k=m} (-1)^{n+1} \binom{1/2}{n}\binom{1-2n}{k} (2p)^{2n-1} (2p-1)^{k} \alpha^{2n+2k-1}\\
                &= \sum_{m=1}^\infty \sum_{n=1}^m (-1)^{n+1} \binom{1/2}{n} \binom{1-2n}{m-n} (2p)^{2n-1} (2p-1)^{m-n} \alpha^{2m-1},
            \end{align*}
            from which the statement follows. Note that
            \[
            \lim_{\alpha\rightarrow 1^-} f(\alpha) = \frac{2p- \sqrt{(2p)^2-4p^2}}{2p}=1,
            \]
            so that $P\left(T_{\underline{0}\rightarrow \overline{0}}<+\infty\right)=1$, like we had anticipated.
        \end{widetext}
    \end{proof}

    To get now the probability of successfully resetting at a particular time, we can just use the formula we just got and compute (as we did in the previous proof for $T_{\underline{0}\rightarrow \overline{-1}}$):
    \[
    P\left( T_{\underline{0}\rightarrow \underline{0}} = t \right) = \sum_{k+l=t} P\left(T_{\underline{0}\rightarrow \overline{0}}=k \right)  P\left(T_{\underline{0}\rightarrow \overline{0}}=l \right).
    \]
    The result is the following formula: 
    \begin{corollary}
        For all natural $m\geq 1$ and all $p\in [0,1]$, with the convention that $0^0=1$,
        \begin{align*}
            &P\left(T_{\underline{0}\rightarrow \underline{0}}=2m\right)=\\
            &= \sum_{k=1}^m \sum_{i=1}^k \sum_{j=1}^{m-k+1} (-1)^{i+j} \binom{1/2}{i} \binom{1/2}{j} \binom{1-2i}{k-i} \cdots \\
            & \cdots \binom{1-2j}{m-k+1-j} (2p)^{2(i+j-1)}(2p-1)^{m+1-(i+j)},\\
            &P\left(T_{\underline{0}\rightarrow \underline{0}}=2m-1\right) =0.
        \end{align*}
    
    \end{corollary}

%\teo{In practice, to implement a trimmed protocol we prepare for the possibility that the systems goes on a superposition of all the possible paths of length $m$ and no more, and then perform this $m$-trimmed protocol twice to get the commutator squared. This means that the probability in Corollary $5$ is not exactly right, but only an upper bound, since it assumes we can go on paths such as (2m-2, wait, 2), which are not available to us if we always measure after $m$ time-steps to determine that the target is in state $[V,W]\ket{\psi}$. The important probability to consider here is therefore the one given by Theorem 3, and plotted in Figure 2 (a). One can just square that to get the probability of success in that case.}

\end{document}